\documentclass[reqno]{amsart}

\usepackage[utf8]{inputenc}
\usepackage[english]{babel} 
\usepackage{amsmath,amssymb,amsfonts,amsthm}
\usepackage{mathrsfs, esint, dsfont}
\usepackage{moreverb,rotating,graphics,float}
\usepackage{caption, subcaption}
\usepackage[colorlinks, linkcolor={blue}, citecolor={red}]{hyperref}
\usepackage{framed,fancybox}
\usepackage{enumerate}
\usepackage{csquotes}

\numberwithin{equation}{section} % For numbering the equations with n°section.
\newtheorem{theorem}{Theorem}[section]
\newtheorem{proposition}[theorem]{Proposition}
\newtheorem{lemma}[theorem]{Lemma}
\newtheorem{corollary}[theorem]{Corollary}

\newtheorem{remark}[theorem]{Remark}
\newtheorem{definition}[theorem]{Definition}

\newtheorem{assumption}[theorem]{Assumption}

\newcommand{\set}[1]{\left\{ #1 \right\}}

\def\C{{\mathbb C}}
\def\CC{{\mathbb C}}

\def\bbI{{\mathbb I}}
\def\KK{{\mathbb K}}
\def\N{{\mathbb N}}

\def\R{{\mathbb R}}
\def\RR{{\mathbb R}}
\def\SS{{\mathbb S}}
\def\TT{{\mathbb T}}
\def\Z{{\mathbb Z}}

\def\re{{\mathrm{e}}} %exponential
\def\ri{{\rm i}} % imaginary i

\def\cA{{\mathcal A}}
\def\cB{{\mathcal B}}
\def\cC{{\mathcal C}}

\def\cG{{\mathcal G}}
\def\cH{{\mathcal H}}

\def\cM{{\mathcal M}}

\def\cS{{\mathcal S}}

\def\cW{{\mathcal W}}

\def\rA{{\rm A}}
\def\rAIII{{\rm AIII}}
\def\rAI{{\rm AI}}
\def\rBDI{{\rm BDI}}
\def\rD{{\rm D}}
\def\rDIII{{\rm DIII}}
\def\rAII{{\rm AII}}
\def\rCII{{\rm CII}}
\def\rC{{\rm C}}
\def\rCI{{\rm CI}}

\def\rX{{\rm X}}

\def\bra{\langle}
\def\ket{\rangle}
\def\Tr{{\rm Tr}}
\def\Ker{{\rm Ker }}
\def\Ran{{\rm Ran }}
\def\Span{{\rm Span }}
\def\dim{{\rm dim }}

\def\U{{\rm U}}
\def\O{{\rm O}}

\def\SU{{\rm SU}}
\def\SO{{\rm SO}}
\def\Sp{{\rm Sp}}
\def\Index{{\rm Index}}

\def\det{{\rm det}}

\def\Pf{{\rm Pf}}

\def\eps{\varepsilon}

\author{David Gontier}
\address[David Gontier]{CEREMADE, University of Paris-Dauphine, PSL University, 75016 Paris, France \& ENS/PSL University, Département de Mathématiques et Applications, F-75005, Paris, France}
\email{gontier@ceremade.dauphine.fr}

\author{Domenico Monaco}
\address[Domenico Monaco]{Dipartimento di Matematica, Sapienza Universit\`a di Roma, Piazzale Aldo Moro 5, 00185 Roma, Italy}
\email{monaco@mat.uniroma1.it}

\author{Solal Perrin-Roussel}
\address[Solal Perrin-Roussel]{ENS Paris-Saclay, 91190 Gif-sur-Yvette, France}
\email{solal.perrin-roussel@ens-paris-saclay.fr}

\title[Symmetric Fermi projections and Kitaev's table]{Symmetric Fermi projections and Kitaev's table: topological phases of matter in low dimensions}
\date{\today}

\begin{document}
    
    \maketitle
    
    \begin{abstract}
        We review Kitaev's celebrated ``periodic table'' for topological phases of condensed matter, which identifies ground states (Fermi projections) of gapped periodic quantum systems up to continuous deformations. We study families of projections which depend on a periodic crystal momentum and respect the symmetries that characterize the various classes of topological insulators. Our aim is to classify such families in a systematic, explicit, and constructive way: we identify numerical indices for all symmetry classes and provide algorithms to deform families of projections whose indices agree. Aiming at simplicity, we illustrate the method for $0$- and $1$-dimensional systems, and recover the (weak and strong) topological invariants proposed by Kitaev and others.
    \end{abstract}

   %\tableofcontents
    
    \section{Introduction}
    
    The goal of the present article is to re-derive the classification of topological phases of quantum matter proposed by Kitaev in his ``periodic table''~\cite{kitaev2009periodic} by means of basic tools from the topology (in particular homotopy theory) of classical groups, and standard factorization results in linear algebra. Kitaev's table classifies ground states of free-fermion systems according to their symmetries and the dimension of the configuration space. We reformulate the classification scheme in terms of homotopy theory, and proceed to investigate the latter issue in dimension $d \in \{0,1\}$. We restrict ourselves to low dimensions in order to illustrate our approach, and to provide constructive proofs for all the classes, using explicit factorization of the matrices that appear. Our intent is thus similar to previous works by Zirnbauer and collaborators \cite{Zirnbauer05,KennedyZirnbauer,KennedyGuggenheim}, which however formulate the notion of free-fermion ground state in a different way, amenable to the investigation of many-body systems. 
    
    The Kitaev classification can be obtained in various ways, using different mathematical tools. Let us mention for instance derivations coming from index theory~\cite{Gro_mann_2015}, K-theory~\cite{Thiang_2015, ProdanSchulzBaldes} and KK-theory~\cite{Bourne_2016}. In addition to this (non-exhaustive) list, one should add the numerous works focusing on one particular case of the table. Our goal here is to provide a short and synthetic derivation of this table, using simple linear algebra.
    
    %%%%%%%%%%%%%%%%%%%%%%%%%%
    
    \subsection{Setting}
    
   Let $\cH$ be a complex Hilbert space of finite dimension $\dim \, \cH \, = N$. For $0 \le n \le N$, $n$-dimensional subspaces of $\cH$ are in one-to-one correspondence with elements of the \emph{Grassmannian}
    \[
    \cG_n(\cH) := \left\{ P \in \cB(\cH) : \: P^2 = P = P^*, \: \Tr(P) = n \right\}
    \] 
    which is comprised of rank-$n$ orthogonal projections in the algebra $\cB(\cH)$ of linear operators on $\cH$. In this paper we are interested in orthogonal-projection-valued continuous functions $P : \TT^d \to \cG_n(\cH)$ which satisfy certain symmetry conditions (to be listed below), and in classifying homotopy classes of such maps. Here $\TT^d = \R^d / \Z^d$ is a $d$-dimensional torus, which we often identify with $[-\tfrac12,\tfrac12]^d$ with periodic boundary conditions. We write $k \ni \TT^d \mapsto P(k) \in \cG_n(\cH)$ for such maps, or  $\{ P(k) \}_{k \in \TT^d}$.
    
    It is well known \cite{Kuchment} that such families of projections arise from the Bloch-Floquet representation of periodic quantum systems on a lattice, in the one-body approximation. In this case, $\cH$ is the Hilbert space accounting for local degrees of freedom in the unit cell associated to the lattice of translations, $\TT^d$ plays the role of the Brillouin torus in (quasi-)momentum space, $k$ is the Bloch (quasi-)momentum, and $P(k)$ is the spectral subspace onto occupied energy levels of some $H(k)$, the Bloch fibers of a periodic lattice Hamiltonian $H$. Two Hamiltonians $H_0$ and $H_1$ are commonly referred to as being in the same topological insulating class if they share the same discrete symmetries (see below) and if they can be continuously deformed one into the other while preserving the symmetries and without closing the spectral gap. This implies that their associated spectral projections $P_0$ and $P_1$ below the spectral gap are homotopic. We thus investigate this homotopy classification directly in terms of projections in momentum space. 
    
    The discrete symmetries that one may want to impose on a family of projections come from those of the underlying quantum-mechanical system. We set the following definitions. Recall that a map $T : \cH \to \cH$ is anti-unitary if it is anti-linear ($T(\lambda x) = \overline{\lambda} T(x)$) and
    \[
        \forall x,y \in \cH, \quad \langle T x, T y \rangle_\cH = \langle y, x \rangle_{\cH} \qquad ( = \overline{ \langle x, y \rangle_\cH }).
    \]
    
    \begin{definition}[Time-reversal symmetry] \label{def:TRS}
        Let $T : \cH \to \cH$ be an anti-unitary operator such that $T^2 =  \eps_T \bbI_{\cH}$ with $\eps_T \in \{-1, 1\}$. We say that a continuous map $P : \TT^d \to \cG_n(\cH)$ satisfies {\bf time-reversal symmetry}, or in short {\bf $T$-symmetry}, if
        \[
        \boxed{ T^{-1} P(k) T = P(-k), \qquad \text{($T$-symmetry)}.}
        \]
        If $\eps_T = 1$, this $T$-symmetry is said to {\bf even}, and if $\eps_T = -1$ it is {\bf odd}. 
    \end{definition}
    
    \begin{definition}[Charge-conjugation/particle-hole symmetry] \label{def:CCS}
        Let $C : \cH \to \cH$ be an anti-unitary operator such that $C^2 =  \eps_C \bbI_{\cH}$ with $\eps_C \in \{-1, 1\}$.
        We say that a continuous map $P : \TT^d \to \cG_n(\cH)$ satisfies {\bf charge-conjugation symmetry} (also called {\bf particle-hole symmetry}), or in short {\bf $C$-symmetry}, if
        \[
        \boxed{ C^{-1} P(k) C = \bbI_{\cH} - P(-k), \quad  \text{($C$-symmetry)}.}
        \]
        If $\eps_C = 1$, this $C$-symmetry is said to {\bf even}, and if $\eps_C = -1$ it is {\bf odd}. 
    \end{definition}
    
    \begin{definition}[Chiral symmetry] \label{def:chiral}
        Let $S : \cH \to \cH$ be a unitary operator such that $S^2 =  \bbI_{\cH}$.
        We say that $P : \TT^d \to \cG_n(\cH)$ satisfies {\bf chiral} or {\bf sublattice symmetry}, or in short {\bf $S$-symmetry}, if
        \[
        \boxed{S^{-1} P(k) S = \bbI_{\cH} - P(k), \quad \text{($S$-symmetry)}.}
        \]
    \end{definition}

    The simultaneous presence of two symmetries implies the presence of the third. In fact, the following assumption is often postulated \cite{Ryu_2010}:
    
    \begin{assumption} \label{S=TC}
        Whenever $T-$ and $C-$ symmetries are both present, we assume that their product $S:=TC$ is an $S$-symmetry, that is, $S$ is unitary and $S^2 = \bbI_\cH$.
    \end{assumption}
    We are not aware of a model in which this assumption is not satisfied, {\em i.e.} in which the $S$ symmetry is unrelated to the $T-$ and $C-$ ones.
    
    \begin{remark} \label{rmk:TC=sigmaCT}
        This assumption is tantamount to require that the operators $T$ and $C$ commute or anti-commute among each other, depending on their even/odd nature. Indeed, the product of two anti-unitary operators is unitary, and the requirement that $S := TC$ satisfies $S^2 = \bbI_\cH$ reads
        \[ 
            TC TC = \bbI_\cH \quad \Longleftrightarrow \quad TC = C^{-1}T^{-1} = \eps_T \eps_C CT. 
        \]
        The same sign determines whether $S$ commutes or anti-commutes with $T$ and $C$. Indeed, we have
        \[ 
        SC = T C^2 = \eps_C T, \quad C S = CTC = T^{-1} C^{-1} C = \eps_T T   \quad \text{so} \quad SC = \eps_T \eps_C CS, 
        \]
        and similarly $ST = \eps_T \eps_C TS$.
    \end{remark}

    %%%%%%%%%%%%%%%%%%%%%%%%%%%%%%%%
    
    Taking into account all possible types of symmetries leads to 10 symmetry classes for maps $P \colon \TT^d \to \cG_n(\cH)$, the famous {\em tenfold way of topological insulators}~\cite{Ryu_2010}. The names of these classes are given in Table~\ref{table:us}, and are taken from the original works of E.~Cartan~\cite{cartan1,cartan2} for the classification of symmetric spaces, which were originally mutuated in \cite{AZ,Zirnbauer05} in the context of random-matrix-valued $\sigma$-models. For a dimension $d \in \N \cup \{ 0 \}$ and a rank $n \in \N$, and for a Cartan label $\rX$ of one of these 10 symmetry classes, we denote by $\rX(d,n,N)$ the set of continuous maps $P \colon \TT^d \to \cG_n(\cH)$, with $\dim(\cH) = N$, and respecting the symmetry requirements of class $\rX$.

    Given two continuous maps $P_0, P_1 \in \rX(d, n ,N)$, we ask the following questions:
    \begin{itemize}
        \item Can we find explicit $\Index \equiv \Index_d^\rX$ maps, which are numerical functions (integer- or integer-mod-2-valued) so that $\Index(P_0) = \Index(P_1)$ iff $P_0$ and $P_1$ are path-connected in $\rX(d, n,N)$?
        \item If so, how to compute this Index?
        \item In the case where ${\rm Index}(P_0) = {\rm Index}(P_1)$, how to construct explicitly a path $P_s$, $s \in [0,1]$ connecting $P_0$ and $P_1$ in $\rX(d, n, N)$?
    \end{itemize}
    In this paper, we answer these questions for all the 10 symmetry classes, and for $d \in \{0, 1\}$. We analyze the classes one by one, often choosing a basis for $\cH$ in which the different symmetry operators $T$, $C$ and $S$ have a specific normal form. In doing so, we recover Cartan's symmetric spaces as $\rX(d=0,n,N)$ -- see the boxed equations in the body of the paper. The topological indices that we find%
    \footnote{We make no claim on the group-homomorphism nature of the Index maps we provide.} %
    are summarized in Table~\ref{table:us}. Our findings agree with the previously mentioned ``periodic tables'' from the physics literature \cite{kitaev2009periodic,Ryu_2010} if one also takes into account the weak $\Z_2$ invariants (see Remark~\ref{rmk:weak}). We note that the $d = 0$ column is not part of the original table. It is related (but not equal) to the $d = 8$ column by Bott periodicity~\cite{Bott_1956}. For our purpose, it is useful to have it explicitly in order to derive the $d = 1$ column.
    
    \begin{table}[ht]
        \centering
        $\begin{array}{| c|ccc|cc|cc |}
            \hline \hline
            \multicolumn{4}{|c|}{\text{Symmetry}} & \multicolumn{2}{c|}{\text{Constraints}} & \multicolumn{2}{c|}{\text{Indices}} \\
            \hline
            \text{Cartan label} & T & C & S & n & N & d=0 & d=1 \\ \hline
            \hyperref[ssec:A]{\rA}    & 0 & 0 & 0 &  &  & 0   & 0    \\
            \hyperref[ssec:AIII]{\rAIII} & 0 & 0 & 1 &  & N=2n & 0    & \Z   \\ \hline
            \hyperref[ssec:AI]{\rAI}   & 1 & 0 & 0 &  &         & 0   & 0    \\
            \hyperref[ssec:BDI]{\rBDI}  & 1 & 1 & 1 &  & N=2n    & \Z_2 & {\Z_2\times\Z}    \\
            \hyperref[ssec:D]{\rD}   & 0 & 1 & 0 &  & N=2n    & \Z_2 & {\Z_2\times\Z_2} \\
            \hyperref[ssec:DIII]{\rDIII} & -1 & 1 & 1 & n=2m \in 2\N & N=2n=4m & 0    & \Z_2 \\
            \hyperref[ssec:AII]{\rAII}  & -1 & 0 & 0 & n=2m \in 2\N & N=2M \in 2\N        & 0  & 0    \\
            \hyperref[ssec:CII]{\rCII}  & -1 & -1 & 1 & n=2m \in 2\N & N=2n=4m & 0    & \Z   \\
            \hyperref[ssec:C]{\rC}    & 0 & -1 & 0 &  & N=2n    & 0    & 0    \\
            \hyperref[ssec:CI]{\rCI}   & 1 & -1 & 1 &  & N=2n    & 0    & 0    \\
            \hline\hline
        \end{array}$
        \medskip
        \caption{A summary of our main results on the topological ``Indices'' of the various symmetry classes of Fermi projections. In the ``Symmetry'' column, we list the sign characterizing the symmetry as even or odd; an entry ``$0$'' means that the symmetry is absent. Some ``Constraints'' may be needed for the symmetry class $\rX(d,n,N)$ to be non-empty.
        }
        \label{table:us}
    \end{table}
    \renewcommand{\arraystretch}{1.0}
    
    \subsection{Notation}
    
    For $\KK \in \{ \RR, \CC\}$, we denote by $\cM_N(\KK)$ the set of $N \times N$ {\em $\KK$-valued} matrices. We denote by $K \equiv K_N : \C^N \to \C^N$ the usual complex conjugation operator. For a complex matrix $A \in \cM_N(\C)$, we set $\overline{A} := K A K$ and $A^T := \overline{A^*}$, where $A^*$ is the adjoint matrix of $A$ for the standard scalar product on $\C^N$.
    
    We then denote by $\cS_N(\KK)$ the set of hermitian matrices ($A = A^*$), and by $\cA_N(\KK)$ the one of skew-hermitian matrices ($A = - A^*$). When $\KK = \CC$, we sometimes drop the notation $\CC$. Also, we denote by $\cS_N^\RR(\CC)$ and $\cA_N^{\RR}(\CC)$ the set of symmetric ($A^T = A$) and antisymmetric matrices ($A^T = -A$).  We denote by $\U(N)$ the subset of unitary matrices,  by $\SU(N)$ the set of unitaries with determinant $1$, by $\O(N)$ the subset of orthogonal matrices, and by $\SO(N)$ the subset of orthogonal matrices with determinant $1$.
    
    We denote by $\bbI\equiv \bbI_N$ the identity matrix of $\C^N$. When $N = 2M$ is even, we also introduce the symplectic matrix
    \[
    J \equiv J_{2M} := \begin{pmatrix} 0 & \bbI_M \\ - \bbI_M & 0 \end{pmatrix}.
    \]
    The symplectic group $\Sp(2M; \KK)$ is defined by
    \begin{equation} \label{eq:Sp}
        \Sp(2M;\KK):= \set{A \in \cM_{2M}(\KK) : A^T J_{2M} A = J_{2M}}.
    \end{equation}
    The {\em compact} symplectic group $\Sp(M)$ is 
    \[
    \Sp(M) := \Sp(2M;\C) \cap \U(2M) = \set{U \in \U(2M) : U^T J_{2M} U = J_{2M}}.
    \]
    
    \subsection{Structure of the paper}
    We study the classes one by one. We begin with the {\em complex classes} A and AIII in Section~\ref{sec:complexClasses}, where no anti-unitary operator is present. We then study non-chiral {\em real} classes (without $S$-symmetry) in Section~\ref{sec:nonchiral}, and chiral classes in Section~\ref{sec:chiral}. In Appendix~\ref{sec:LA}, we review some factorizations of matrices, which allow us to prove our results.

    %%%%%%%%%%%%%%%%%%%%%%%%%%%%%%%%%%%%

    \section{Complex classes: $\rA$ and $\rAIII$}
    \label{sec:complexClasses}
    
    The symmetry classes $\rA$ and $\rAIII$ are often dubbed as \emph{complex}, since they do not involve any antiunitary symmetry operator, and thus any ``real structure'' induced by the complex conjugation. By contrast, the other 8 symmetry classes are called \emph{real}. Complex classes where studied, for example, in~\cite{ProdanSchulzBaldes,denittis2018chiral}.
    
    \subsection{Class $\rA$}
    \label{ssec:A}
    
    In class $\rA$, no discrete symmetry is imposed. We have in this case
    
    \begin{theorem}[Class $\rA$] \label{thm:A}
        The sets $\rA(0, n, N)$ and $\rA(1, n, N)$ are path-connected.
    \end{theorem}
    \begin{proof}
        Since no symmetry is imposed and $\TT^0 = \{0\}$ consists of a single point, we have $\rA(0,n,N) = \cG_n(\cH)$. It is known~\cite[Ch.~8, Thm.~2.2]{husemoller} that the complex Grassmannian is connected, hence so is $\rA(0,n,N)$. This property follows from the fact that the map $\U(N) \to \cG_n(\cH)$ which to any $N \times N$ unitary matrix associates the linear span of its first $n$ columns (say in the canonical basis for $\cH \simeq \C^N$), viewed as orthonormal vectors in $\cH$, induces a bijection
        \begin{equation} \label{eq:U/UxU=G}
            \boxed{ \rA(0,n,N) \simeq \cG_n(\cH) \simeq \U(N) / \U(n) \times \U(N-n).}
        \end{equation}
        Since $U(N)$ is connected, so is $\rA(0, n, N)$. 
        
        \medskip
        
        To realize this explicitly, we fix a basis of $\cH \simeq \C^N$. Let $P_0, P_1 \in \rA(0, n, N)$. For $j \in \{ 0, 1\}$, we choose a unitary $U_j \in \U(N)$ such that its $n$ first column vectors span the range of $P_j$. We then choose a self-adjoint matrix $A_j \in \cS_N$ so that $U_j = \re^{ \ri A_j}$. We now set, for $s \in (0, 1)$, 
        \[
        U_s := \re^{ \ri A_s}, \quad A_s := (1 - s) A_0 + s A_1.
        \]
        The map $s \mapsto U_s$ is continuous, takes values in $\U(n)$, and connects $U_0$ and $U_1$. The projection $P_s$ on the first $n$ column vectors of $U_s$ then connects $P_0$ and $P_1$, as wanted.
        
        \medskip
        
        We now prove our statement concerning $\rA(1, n, N)$. Let $P_0, P_1 \colon \TT^1 \to \cG_n(\cH)$ be two periodic families of projections. Recall that we identify $\TT^1 \simeq [-1/2, 1/2]$. Consider the two projections $P_0(-\tfrac12) = P_0(\tfrac12)$ and $P_1(-\tfrac12) = P_1(\tfrac12)$, and connect them by some continuous path $P_s(-\frac12) = P_s(\frac12)$ as previously. The families $P_0(k)$ and $P_1(k)$, together with the maps $P_s(-\tfrac12)$ and $P_s(\tfrac12)$, define a continuous family of projectors on the boundary $\partial \Omega$ of the square
        \begin{equation} \label{eq:Omega}
            \Omega := [-\tfrac12, \tfrac12] \times [0,1] \ni (k,s).
        \end{equation}
        It is a standard result (see for instance~\cite[Lemma 3.2]{Gontier2019numerical} for a constructive proof) that such families can be extended continuously to the whole set $\Omega$. This gives an homotopy $P_s(k) = P(k,s)$ between $P_0$ and $P_1$.
    \end{proof}
    
    \subsection{Class $\rAIII$}
    \label{ssec:AIII}
    
    In class $\rAIII$, only the $S$-symmetry is present. It is convenient to choose a basis in which $S$ is diagonal. This is possible thanks to the following Lemma, which we will use several times in classes where the $S$-symmetry is present.
    \begin{lemma} \label{lem:formS}
        Assume $\rAIII(d=0, n, N)$ is non-empty. Then $N = 2n$, and there is a basis of $\cH$ in which $S$ has the block-matrix form
        \begin{equation} \label{eq:formS}
            S = \begin{pmatrix} \bbI_n & 0 \\ 0 & - \bbI_n \end{pmatrix}.
        \end{equation}
        In this basis, a projection $P$ satisfies $S^{-1} P S = \bbI_\cH - P$ iff it has the matrix form
        \begin{equation} \label{eq:special_form_P}
            P = \frac12  \begin{pmatrix}
                \bbI_n & Q \\ Q^* & \bbI_n
            \end{pmatrix} \quad \text{with} \quad Q \in \U(n).
        \end{equation}
    \end{lemma}
    
    \begin{proof}
        Let $P_0 \in \rAIII(0, n, N)$. Since $S^{-1} P_0 S =  \bbI_{\cH} - P_0$, $P_0$ is unitarily equivalent to $\bbI_\cH - P_0$, hence $\cH = {\rm Ran} \ P_0 \oplus  {\rm Ran} \ (\bbI_{\cH} - P_0)$ is of dimension $N= 2n$. \\
        Let $(\psi_1, \psi_2, \cdots, \psi_n)$ be an orthonormal basis for ${\rm Ran } \, P_0$. We set
        \[
        \forall i \in \{ 1, \cdots, n\}, \quad \phi_i := \frac{1}{\sqrt{2}}(\psi_i + S \psi_i), \quad \phi_{n+i} = \frac{1}{\sqrt{2}}(\psi_i - S \psi_i).
        \] 
        The family $(\phi_1, \cdots, \phi_{2n})$ is an orthonormal basis of $\cH$, and in this basis, $S$ has the matrix form~\eqref{eq:formS}.
        
        \medskip
        
        For the second point, let $P \in \rAIII(0,n,2n)$, and decompose $P$ in blocks:
        \[ P = \frac12 \begin{pmatrix}
            P_{11} & P_{12} \\
            P_{12}^* & P_{22}
        \end{pmatrix}.
        \]
        The equation $S^{-1} P S= \bbI_\cH - P$ implies that $P_{11} = P_{22} = \bbI_n$. Then, the equation $P^2 = P$ shows that $P_{12} =: Q$ is unitary, and~\eqref{eq:special_form_P} follows. 
    \end{proof}
    
    The previous Lemma establishes a bijection $P \longleftrightarrow Q$, that is
    \[ 
    \boxed{ \rAIII(0,n,2n) \simeq \U(n).} 
    \]
    For $P \in \rAIII(d, n, 2n)$, we denote by $Q : \TT^d \to \U(n)$ the corresponding periodic family of unitaries. 
    
    For a curve $\cC$ homeomorphic to $\SS^1$, and for $Q : \cC \to \U(n)$, we denote by ${\rm Winding}(\cC, Q)$ the usual winding number of the determinant of $Q$ along $\cC$.
    
    %-----------------------------------------------------------------------
    \begin{theorem}[Class $\rAIII$]
        The set $\rAIII(d, n, N)$ is non-empty iff $N=2n$. 
        \begin{itemize}
            \item The set $\rAIII(0,n,2n)$ is path-connected.
            \item Define the index map $\Index_1^{\rAIII} : \rAIII(1, n, 2n) \to \Z$ by
            \[
            \forall P \in \rAIII(1, n, 2n), \quad \Index_1^{\rAIII} (P) := {\rm Winding}(\TT^1, Q).
            \]
            Then $P_0$ is homotopic to $P_1$ in $\rAIII(1,n,2n)$ iff $\Index_1^{\rAIII}(P_0) = \Index_1^{\rAIII}(P_1)$.
        \end{itemize}
    \end{theorem}
    %-----------------------------------------------------------------------
    
    \begin{proof}
        We already proved that $N = 2n$. Since $\U(n)$ is connected, so is $\rAIII(0, n, 2n)$. A constructive path can be constructed as in the previous section using exponential maps.
        
        We now focus on $\rAIII(d = 1, n, 2n)$. Analogously, the question of whether two maps in $\rAIII(1,n,2n)$ are continuously connected by a path can be translated in whether two unitary-valued maps $Q_0, Q_1 \colon \TT^1 \to \U(n)$ are homotopic to each other. As in the previous proof, consider the unitaries $Q_0(-\tfrac12) = Q_0(\tfrac12) \in \U(n)$ and $Q_1(-\tfrac12) = Q_1(\tfrac12) \in \U(n)$. Connect them by some $Q_s(-\frac12) = Q_s(\frac12)$ in $\U(n)$. This defines a $\U(n)$-valued map on $\partial \Omega$, where the square $\Omega$ is defined in~\eqref{eq:Omega}.
        
        It is well known that one can extend such a family of unitaries to the whole $\Omega$ iff ${\rm Winding}(\partial \Omega,Q) = 0$ (see~\cite[Section IV.B]{Gontier2019numerical} for a proof, together with a constructive proof of the extension in the case where the winding vanishes). In our case, due to the orientation of the boundary of $\Omega$ and of the periodicity of $Q_0(k), Q_1(k)$, we have
        \[ {\rm Winding}(\partial \Omega, Q) =  {\rm Winding}(\TT^1,\ Q_1) -  {\rm Winding}(\TT^1, Q_0), \]
        which is independent of the previously constructed path $Q_s(\frac12)$. The conclusion follows.
    \end{proof}

    %%%%%%%%%%%%%%%%%%%%%%%%%%%%%%%%%%%%
    
    \section{Real non-chiral classes: $\rAI$, $\rAII$, $\rC$ and $\rD$}
    \label{sec:nonchiral}
    
    Next we consider those symmetry classes which are characterized by the presence of a \emph{single} anti-unitary symmetry: a $T$-symmetry (which even in class $\rAI$ and odd in class $\rAII$) or a $C$-symmetry (whih is even in class $\rD$ and odd in class $\rC$). In particular, these classes involve anti-unitarily intertwining $P(k)$ and $P(-k)$. For these symmetry classes, the analysis of their path-connected components in dimension $d=1$ is reduced to that of dimension $d=0$, thanks to the following Lemma.
    
    %--------------------------------------------
    \begin{lemma}[Real non-chiral classes in $d=1$] \label{lem:NonCh1d}
        Let $\rX \in \set{\rAI, \rAII, \rC, \rD}$. Then $P_0$ and $P_1$ are in the same connected component of $\rX(1,n,N)$ iff
        \begin{itemize}
            \item $P_0(0)$ and $P_1(0)$ are in the same connected component in $\rX(0,n,N)$, and
            \item $P_0(\tfrac12)$ and $P_1(\tfrac12)$ are in the same connected component in $\rX(0,n,N)$.
        \end{itemize}
    \end{lemma}
    %--------------------------------------------
    
    \begin{proof}
        
        We give the argument for the class $\rX = \rD$, but the proof is similar for the other classes. First, we note that if $P_s(k)$ connects $P_0$ and $P_1$ in $\rD(1, n, N)$, then for $k_0 \in \{ 0, \tfrac12\}$ one must have $C^{-1} P_s(k_0) C = \bbI_\cH - P_s(k_0)$, so $P_s(k_0)$ connects $P_0(k_0)$ and $P_1(k_0)$ in $\rD(d = 0, n, N)$.
        
        Let us prove the converse. Assume that $P_0$ and $P_1$ are two projection-valued maps in $\rD(1, n, N)$ so that there exist paths $P_s(k_0)$ connecting $P_0(k_0)$ and $P_1(k_0)$ in $\rD(0,n,N)$, for the high symmetry points $k_0 \in \{ 0, \tfrac12\}$. Denote by $\Omega_0$ the half-square
        \begin{equation} \label{eq:Omega0}
            \Omega_0 := [0, \tfrac12] \times [0, 1] \quad \ni (k,s),
        \end{equation}
        (compare with~\eqref{eq:Omega}). The families 
        \[
            \set{P_0(k)}_{k \in [0,1/2]}, \quad \set{P_1(k)}_{k\in[0,1/2]}, \quad \set{P_s(0)}_{s\in[0,1]}
            \quad \text{and} \quad \set{P_s(\tfrac12)}_{s\in[0,1]},
        \]
     together define a continuous family of projectors on the boundary $\partial \Omega_0$. As was already mentioned in Section~\ref{ssec:A}, this family can be extended continuously on the whole set $\Omega_0$. 
        
        This gives a continuous family $\set{P_s(k)}_{k\in[0,1/2],\,s\in[0,1]}$ which connects continuously the restrictions of $P_0$ and $P_1$ to the half-torus $k \in [0, \tfrac12]$. We can then extend the family of projections to $k \in [-\tfrac12, 0]$ by setting
        \[ 
            \forall k \in [- \tfrac12, 0], \ \forall s \in [0, 1], \quad P_s(k) := C \big[ \bbI_\cH - P_s (-k) \big] C^{-1}.
         \]
        By construction, for all $s \in [0, 1]$, the map $P_s$ is in $\rD(1,n,N)$. In addition, since at $k_0 \in \{0 ,\tfrac12\}$ we have $P_s(k_0) \in \rD(0, n, N)$, the above extension is indeed continuous as a function of $k$ on the whole torus $\TT^1$. This concludes the proof. 
    \end{proof}
    
    \subsection{Class $\rAI$} \label{ssec:AI}
    
    In class $\rAI$, the relevant symmetry is an anti-unitary operator $T$ with $T^2 = \bbI_\cH$. This case was studied for instance in~\cite{panati2007triviality,denittis2014real,FiorenzaMonacoPanatiAI}.
    
    \begin{lemma} \label{lem:formTeven}
        If $T$ is an anti-unitary operator on $\cH$ such that $T^2 = \bbI_\cH$, then there is a basis of $\cH$ in which $T$ has the matrix form $T = K_N$.
    \end{lemma}
    
    \begin{proof}
        We construct the basis by induction. Let $\psi_1 \in \cH$ be a normalized vector. if $T \psi_1 = \psi_1$, we set $\phi_1 = \psi_1$, otherwise we set 
        \[
        \phi_1 := \ri \dfrac{ \psi_1 - T \psi_1 }{\| \psi_1 - T \psi_1 \|}.
        \]
        In both cases, we have $T \phi_1 = \phi_1$ and $\| \phi_1 \| = 1$, which gives our first vector of the basis. Now take $\psi_2$ orthogonal to $\phi_1$. We define $\phi_2$ as before. If $\phi_2 = \psi_2$, then $\phi_2$ is automatically orthogonal to $\phi_1$. This also holds in the second case, since 
        \[
        \langle \psi_2 - T \psi_2, \phi_1 \rangle 
        = - \langle T \psi_2, \phi_1 \rangle 
        = - \langle T \phi_1, T^2 \psi_2 \rangle 
        = - \langle  \phi_1, \psi_2 \rangle = 0,
        \]
        where we used twice that $\langle \psi_2, \phi_1 \rangle = 0$. We go on, and construct the vectors $\phi_k$ inductively for $1 \le k \le N$. This gives an orthonormal basis in which $T = K$.
    \end{proof}

    \begin{theorem}[Class $\rAI$] \label{thm:AI}
        The sets $\rAI(0, n, N)$ and $\rAI(1, n, N)$ are path-connected.
    \end{theorem}
    \begin{proof}
        In a basis in which $T = K_N$, we have the identification
        \[ \rAI(0, n, N) = \left\{ P \in \cG_n(\C^N) : \overline{P} = P \right\}. \]
        In other words, $\rAI(0, n, N)$ consists of \emph{real} subspaces of $\cH$, {\em i.e.} those that are fixed by the complex conjugation $T=K$. One can therefore span such subspaces (as well as their orthogonal complement) by orthonormal \emph{real} vectors. This realizes a bijection similar to~\eqref{eq:U/UxU=G}, but where unitary matrices are replaced by orthogonal ones: more precisely
        \[ 
        \boxed{ \rAI(0, n, N) \simeq \O(N) / \O(n) \times \O(N-n).}
        \]
%        {\bf ***David***: is it true? The other vectors of the basis need not be real-valued...}
        
        We adapt the argument in the proof of Theorem~\ref{thm:A} to show that the latter space is path-connected. Let $P_0, P_1 \in \rAI(0, n, N)$. We choose two {\em real} bases of $\cH$, which we identify with columns of orthogonal matrices $U_0,U_1 \in \O(N)$, so that the first $n$ vectors of $U_j$ span the range of $P_j$, for $j \in \set{0,1}$. In addition, by flipping the first vector, we may assume $U_0, U_1 \in \SO(N)$. Then there is $A_0, A_1 \in \cA_N(\R)$ so that $U_j = \re^{ A_j}$ for $j \in \{ 0, 1 \}$. We then set $U_s := \re^{A_s}$ with $A_s = (1 - s) A_0 + sA_1$. The projection $P_s$ on the first $n$ column vectors of $U_s$ then interpolates between $P_0$ and $P_1$, as required. In view of Lemma~\ref{lem:NonCh1d}, the path-connectedness of $\rAI(0,n,N)$ implies the one of $\rAI(1,n,N)$.
    \end{proof}
    
    \subsection{Class $\rAII$}
    \label{ssec:AII}
    
    In class $\rAII$ we have $T^2 = -\bbI_\cH$. This case was studied for instance in~\cite{graf2013bulk, denittis2015quaternionic, FiorenzaMonacoPanatiAII, cornean2017wannier, monaco2017gauge}.
    
    \begin{lemma} \label{lem:formTodd}
        There is an anti-unitary map $T : \cH \to \cH$ with $T^2 = - \bbI_\cH$ iff ${\rm dim } \, \cH = N = 2M$ is even. In this case, there is a basis of $\cH$ in which $T$ has the matrix form
        \begin{equation} \label{eq:oddT}
            T = \begin{pmatrix}
                0 & K_M \\ - K_M & 0
            \end{pmatrix} = J_{2M} \, K_{2M}.
        \end{equation}
    \end{lemma}
    
    \begin{proof}
        First, we note that $T \psi$ is always orthogonal to $\psi$. Indeed, we have
        \begin{equation} \label{eq:psiTpsi}
            \langle \psi, T \psi \rangle =  \langle T^2 \psi, T \psi \rangle = - \langle \psi, T \psi \rangle, \quad \text{hence} \quad
            \langle \psi, T \psi \rangle = 0.
        \end{equation}
        We follow the strategy employed {\em e.g.}~in~\cite{graf2013bulk} and~\cite[Chapter 4.1]{cornean2017wannier}, and construct the basis by induction. Let $\psi_1 \in \cH$ be any normalized vector, and set $\psi_2 := T \psi_1$. The family $\{\psi_1, \psi_2\}$ is orthonormal by~\eqref{eq:psiTpsi}. If $\cH \neq \Span \{ \psi_1, \psi_2\}$, then there is $\psi_3 \in \cH$ orthonormal to this family. We then set $\psi_4 = T \psi_3$, and claim that $\psi_4$ is orthonormal to the family $\{ \psi_1, \psi_2, \psi_3\}$. First, by~\eqref{eq:psiTpsi}, we have $\bra \psi_3, \psi_4 \ket = 0$. In addition, we have 
        \[
        \bra \psi_4, \psi_1 \ket = \bra T \psi_3, \psi_1 \ket = \bra T \psi_1, T^2 \psi_3 \ket = - \bra \psi_2, \psi_3 \ket = 0,
        \]
        and, similarly,
        \[
        \bra \psi_4, \psi_2 \ket = \bra T \psi_3, T \psi_1 \ket = \bra T^2 \psi_1, T^2 \psi_3 \ket = \bra \psi_1, \psi_3 \ket = 0.
        \]
        We proceed by induction. We first obtain that the dimension of $\cH$ is even, $N = 2M$, and we construct an explicit basis $\{\psi_1, \cdots, \psi_{2M}\}$ for $\cH$. In the orthonormal basis $\{\psi_1, \psi_3, \psi_5, \cdots, \psi_{2M-1}, \psi_2, \psi_4, \cdots \psi_{2M}\}$, the operator $T$ has the matrix form~\eqref{eq:oddT}.
    \end{proof}

    \begin{theorem}[Class $\rAII$] \label{thm:AII}
        The sets $\rAII(0, n, N)$ and $\rAII(1, n, N)$ are non-empty iff $n = 2m \in 2\N$ and $N = 2M \in 2\N$. Both are path-connected.
    \end{theorem}
    
    \begin{proof}
        The proof follows the same lines as that of Theorems~\ref{thm:A} and~\ref{thm:AI}. The condition $T^{-1} P T = P$ for $P \in \rAII(0, n, N)$ means that the range of the projection $P$ is stable under the action of $T$. This time, the operator $T$ endows the Hilbert space $\cH$ with a \emph{quaternionic} structure, namely the matrices $\set{\ri \bbI_\cH, T, \ri T}$ satisfy the same algebraic relations as the basic quaternions $\set{\mathbf{i},\mathbf{j},\mathbf{k}}$: they square to $-\bbI_\cH$, they pairwise anticommute and the product of two successive ones cyclically gives the third. This allows to realize the class $\rAII(0,2m,2M)$ as
        \[ 
        \boxed{ \rAII(0, 2m, 2M) \simeq \Sp(M) / \Sp(m) \times \Sp(M-m). }
        \]
        Matrices in $\Sp(M)$ are exponentials of Hamiltonian matrices, that is, matrices $A$ such that $J_{2M} A$ is symmetric \cite[Prop.~3.5 and Coroll.~11.10]{hall2015lie}. Such matrices form a (Lie) algebra, and therefore the same argument as in the proof of Theorem~\ref{thm:AI} applies, yielding path-connectedness of $\rAII(0,2m,2M)$. This in turn implies, in combination with Lemma~\ref{lem:NonCh1d}, that $\rAII(1,2m,2M)$ is path-connected as well.
    \end{proof}

    \subsection{Class $\rD$}
    \label{ssec:D}
    
    We now come to classes where the $C$-symmetry is present.  We first focus on the even case, $C^2 = + \bbI_\cH$, characterizing class $\rD$. One of the most famous models in this class is the 1-dimensional Kitaev chain~\cite{KitaevChain}. We choose to work in the basis of $\cH$ in which $C$ has the form%
    \footnote{This is \emph{different} from the ``energy basis'', of common use in the physics literature, in which $C$ is block-off-diagonal, mapping ``particles'' to ``holes'' and vice-versa. We find this other basis more convenient for our purpose.}  $C = K_N$ (see Lemma~\ref{lem:formTeven}). 
    
    \begin{lemma}
        The set $\rD(0, n, N)$ is non-empty iff $N = 2n$. In this case, and in a basis where $C = K_N$, a projection $P$ is in $\rD(0, n,2n)$ iff it has the matrix form
        \[
        P = \frac12 (\bbI_N + \ri A ), \quad \text{with} \quad A \in \O(2n) \cap \cA_{2n}(\R).
        \]
    \end{lemma}
    
    \begin{proof}
        A computation shows that
        \[
        \begin{cases}
            P^* = P \\
            P^2 = P \\
            C^{-1}  P C = \bbI - P
        \end{cases}
        \Longleftrightarrow \quad
        \begin{cases}
            A^* = -A \\
            A^2 = - \bbI_N \\
            \overline{A}  = A
        \end{cases}
        \Longleftrightarrow \quad
        \begin{cases}
            A^* A = \bbI_N \\
            A= \overline{A} = - A^T.
        \end{cases}
        \]
        This proves that $P \in \rD(0, n, N)$ iff $A \in \O(N) \cap \cA_{N}(\R)$. In particular, we have $\det(A) = (-1)^N \det(-A) = (-1)^N \det(A^T) = (-1)^N \det(A)$, so $N = 2m$ is even. Finally, since the diagonal of $A$ is null, we have $n = \Tr(P) = \frac12 \Tr(\bbI_N) = m$.
    \end{proof}
    In Corollary~\ref{cor:O(d)capA(d)} below, we prove that a matrix $A$ is in $\O(2n) \cap \cA_{2n}(\R)$ iff it is of the form
     \[ A = W^T J_{2n} W, \quad \text{with} \quad W \in \O(2n). \]
     In addition, we have $W_0^T J_{2n} W_0 = W_1^T J_{2n} W_1$ with $W_0, W_1 \in \O(2n)$ iff $W_0 W_1^* \in \Sp(n) \cap \O(2n)$. Finally, in Proposition~\ref{prop:SpO=U}, we show that $\Sp(n) \cap \O(2n) \simeq \U(n)$. Altogether, this shows that
    \[
    \boxed{ \rD(0, n, 2n) \simeq  \O(2n) \cap \cA_{2n}(\R) \simeq \O(2n)/\U(n) .}
    \]
    
    To identify the connected components of this class, recall that for an anti-symmetric matrix $A \in \cA_{2n}^\R(\C)$, we can define its \emph{Pfaffian}
    \begin{equation} \label{eq:def:Pfaffian}
        \Pf(A)   := \dfrac{1}{2^n n!} \sum_{\sigma} {\rm sgn}(\sigma) \prod_{i=1}^n a_{\sigma(2i-1), \sigma(2i)},
    \end{equation}
    where the above sum runs over all permutations over $2n$ labels and ${\rm sgn}(\sigma)$ is the sign of the permutation $\sigma$. The Pfaffian satisfies
    \[ \Pf(A)^2 = \det(A). \]
    On the other hand, if $A \in \O(2n)$, then $\det(A) \in \{ \pm 1\}$, so if $A \in \O(2n) \cap \cA_{2n}(\R)$, we must have $\det(A) = 1$ and $\Pf(A) \in \{ \pm 1\}$.

    \begin{theorem}[Class $\rD$]
        The set $\rD(d, n, N)$ is non-empty iff $N=2n$. 
        \begin{itemize}
            \item The set $\rD(0, n, 2n)$ has two connected components. Define the index map $\Index_0^{\rD} \colon \rD(0,n,2n) \to \Z_2 \simeq \{ \pm 1\}$ by
            \[
            \forall P \in \rD(0, n, 2n), \quad \Index_0^\rD (P) := \Pf(A).
            \]
            Then $P_0$ is homotopic to $P_1$ in  $\rD(0,n,2n)$ iff $ \Index_0^{\rD}(P_0) = \Index_0^{\rD}(P_1)$.
            \item The set  $\rD(1, n, 2n)$ has four connected components. Define the index map $\Index_1^{\rD} \colon \rD(1,n,2n) \to \Z_2 \times \Z_2$ by
            \[
            \forall P \in \rD(1, n, 2n), \quad \Index_1^\rD (P) := \left( \Pf(A(0)), \Pf(A(\tfrac12))\right).
            \]
            Then $P_0$ is homotopic to $P_1$ in  $\rD(1,n,2n)$ iff $ \Index_1^{\rD}(P_0) = \Index_1^{\rD}(P_1)$.
        \end{itemize}
    \end{theorem}

    \begin{proof}
        We start with $\rD(0, n, N=2n)$. Let $P_0, P_1 \in \rD(0, n, 2n)$. It is clear that if $\Pf(A_0) \neq \Pf(A_1)$, then $P_0$ and $P_1$ are in two different connected components (recall that $\Pf(\cdot)$ is a continuous map, with values in $\{ \pm 1\}$ in our case).
        
        It remains to construct an explicit homotopy in the case where $\Pf(A_0) = \Pf(A_1)$. In Corollary~\ref{cor:O(d)capA(d)} below, we recall that a matrix $A$ is in $\O(2n) \cap \cA_{2n}(\R)$ iff there is $V \in \SO(2n)$ so that
        \[
        A = V^T D V, \quad \text{with} \quad D = (1, 1, \cdots, 1, \Pf(A)) \otimes \begin{pmatrix}
            0 & 1 \\ - 1 & 0
        \end{pmatrix}.
        \]
        So, if $A_0, A_1 \in \O(2n) \cap \cA_{2n}(\R)$ have the same Pfaffian, it is enough to connect the corresponding $V_0$ and $V_1$ in $\SO(2n)$. The proof follows since $\SO(2n)$ is path-connected (compare with the proof of Theorem~\ref{thm:AI}).
        
        The case for $\rD(d = 1, n, 2n)$ is now a consequence of Lemma~\ref{lem:NonCh1d}.
    \end{proof}
    
    \begin{remark} \label{rmk:weak}
    For 1-dimensional translation-invariant systems, one can distinguish between a \emph{weak} ({\em i.e.}, lower-dimensional, depending solely on $P(k)$ at $k = 0$) index
    \[ \Index_0^\rD(P(0)) = \Pf(A(0)) \in \Z_2 \]
    and a \emph{strong} (i.e., ``truly'' 1-dimensional) index
    \[ \widetilde{ \Index_0^\rD } (P) := \Pf(A(0)) \cdot \Pf(A(\tfrac12)) \in \Z_2. \]
    Only the latter $\Z_2$-index appears in the periodic tables for free ground states \cite{kitaev2009periodic}. Our proposed index
    \[ \Index_1^\rD (P) = \left( \Pf(A(0)), \Pf(A(\tfrac12))\right) \in \Z_2 \times \Z_2 \]
    clearly contains the same topological information of both the weak and strong indices.
    
    A similar situation will appear in class $\rBDI$ (see Section~\ref{ssec:BDI}).
    \end{remark}

    \subsection{Class $\rC$}
    \label{ssec:C}
    
    We now focus on the odd $C$-symmetry class, where $C^2 = - \bbI_\cH$. Thanks to Lemma~\ref{lem:formTodd}, $N = 2M$ is even, and we can choose a basis of $\cH$ in which $C$ has the matrix form
    \[
    C = \begin{pmatrix}
        0 & K_M \\ - K_M & 0
    \end{pmatrix} = J_{2M}\, K_{2M}.
    \]
    Recall that $\Sp(n) := \Sp(2n; \C) \cap \U(2n)$.
    \begin{lemma}
        The set $\rC(0, n, N)$ is non-empty iff $N = 2n$ (hence $n = M$). A projection $P$ is in $\rC(0, n, 2n)$ iff it has the matrix form
        \[
        P = \frac12 \left( \bbI_{2n} + \ri J_{2n}A  \right), \quad \text{with} \quad A \in \Sp(n) \cap \cS^\R_{2n}(\C).
        \]
    \end{lemma}
    \begin{proof}
        With this change of variable, we obtain that
        \[
        \begin{cases}
            P = P^* \\
            P^2 = P \\
            C^{-1} P C = \bbI_{2n} - P
        \end{cases}
        \Longleftrightarrow \quad
        \begin{cases}
            A^* J_{2n} = J_{2n} A\\
            J_{2n} A J_{2n} A = - \bbI_{2n} \\
            \overline{A} J_{2n} = J_{2n} A.
        \end{cases}
        \]
        With the two first equations, we obtain $A A^* = \bbI_{2n}$, so $A \in \U(2n)$. With the first and third equations, we get $A^T = A$, so $A \in \cS_{2n}^\R(\C)$, and with the two last equations, $A^T J_{2n} A = J_{2n}$, so $A \in \Sp(2n; \C)$. The result follows.
    \end{proof}
    In Corollary~\ref{cor:AutonneTagaki} below, we prove that a matrix $A$ is in $\Sp(n) \cap \cS_{2n}^\R(\C)$ iff it is of the form
    \[ 
    A = V^T V, \quad \text{for some} \quad V \in \Sp(n).
    \]
    In addition, $A = V_0^T V_0 = V_1^T V_1$ with $V_0, V_1 \in \Sp(n)$ iff $V_1 V_0^* \in \Sp(n) \cap \O(2n) \simeq \U(n)$ (see the already mentioned Proposition~\ref{prop:SpO=U} for the last bijection). This proves that
    \[
    \boxed{ \rC(0, n, N) \simeq \Sp(n) \cap \cS^\R_{2n}(\C) \simeq \Sp(n) / \U(n). }
    \]
    
    \begin{theorem}[Class $\rC$]
        The sets $\rC(0, n, N)$ and $\rC(1, n, N)$ are non-empty iff $N = 2n$. Both are path-connected.
    \end{theorem}

    \begin{proof}
        For $\rC(d=0, n, 2N)$, it is enough to prove that $\Sp(n) \cap \cS_{2n}^\R(\C)$ is path-connected. 
        To connect $A_0$ and $A_1$ in $\Sp(n) \cap \cS_{2n}(\C)$ it suffices to connect the corresponding $V_0$ and $V_1$ in $\Sp(2n)$. This can be done as we already saw in the proof of Theorem~\ref{thm:AII}. Invoking Lemma~\ref{lem:NonCh1d} allows to conclude that $\rC(1,n,2n)$ is path-connected as well.
    \end{proof}

    %%%%%%%%%%%%%%%%%%%%%%%%%%%%%%%%%%%%
    
    \section{Real chiral classes: $\rBDI$, $\rDIII$, $\rCII$ and $\rCI$}
    \label{sec:chiral}
    
    We now focus on the chiral real classes; by Assumption~\ref{S=TC}, the chiral symmetry operator $S$ will come from the combination of a $T$-symmetry with a $C$-symmetry. In what follows, we will always find a basis for $\cH$ in which $S:=TC$ has the form~\eqref{eq:formS}. In particular, Lemma~\ref{lem:formS} applies, and any $P \in \rX(d,n,2n)$ for $\rX \in \set{\rBDI, \rDIII, \rCII,\rCI}$ will be of the form
    \begin{equation} \label{eq:Q}
       P(k) = \frac12 \begin{pmatrix} \bbI_n & Q(k) \\ Q(k)^* & \bbI_n \end{pmatrix} \quad \text{with} \quad Q(k) \in \U(n).
    \end{equation}
    The $T$-symmetry (or equivalently the $C$-symmetry) of $P(k)$ translates into a condition for $Q(k)$, of the form
    \begin{equation} \label{eq:FTQ}
        F_T(Q(k)) = Q(-k).
    \end{equation}
    With these remarks, we are able to formulate the analogue of Lemma~\ref{lem:NonCh1d} for real chiral classes.
    
    \begin{lemma}[Real chiral classes in $d=1$] \label{lem:Ch1d}
        Let $\rX \in \set{\rBDI, \rDIII, \rCII,\rCI}$. Then $P_0$ and $P_1$ are in the same connected component in $\rX(1,n,2n)$ iff
        \begin{itemize}
            \item $P_0(0)$ and $P_1(0)$ are in the same connected component in $\rX(0,n,2n)$, 
            \item $P_0(\tfrac12)$ and $P_1(\tfrac12)$ are in the same connected component in $\rX(0,n,2n)$, and
            \item there exists a choice of the above interpolations $P_s(0)$, $P_s(\tfrac12)$, $s \in [0,1]$, and therefore of the corresponding unitaries $Q_s(0)$, $Q_s(\tfrac12)$ as in~\eqref{eq:Q}, such that 
            \[ {\rm Winding} (\partial\Omega_0, Q)=0, \]
            where $\Omega_0$ is the half-square defined in~\eqref{eq:Omega0}, and where $Q$ is the continuous family of unitaries defined on $\partial \Omega_0$ via the families
            \[
                \left\{ Q_0(k)\right\}_{k\in [0, 1/2]}, \quad
                 \left\{ Q_1(k)\right\}_{k\in [0, 1/2]}, \quad
                  \left\{ Q_s(0)\right\}_{s\in [0, 1]}, \quad \text{and} \quad
                   \left\{ Q_s(\tfrac12)\right\}_{s\in [0, 1]}.
            \]
        \end{itemize}
    \end{lemma}
    \begin{proof}
        As was already mentioned, the vanishing of the winding in the statement is equivalent to the existence of a continuous extension of the map $Q(k,s) \equiv Q_s(k)$ to $(k,s) \in \Omega_0$. For $k \in [-\tfrac12,0]$ and $s \in [0,1]$, we define 
        \[ Q_s(k) := F_T(Q_s(-k)), \]
        where $F_T$ is the functional relation in~\eqref{eq:FTQ}. Using~\eqref{eq:Q}, we can infer the existence of a family of projections $\set{P_s(k)}_{k \in \TT^1}$ which depends continuously on $s \in [0,1]$, is in $\rX(1,n,2n)$ for all $s \in [0,1]$, and restricts to $P_0$ and $P_1$ at $s=0$ and $s=1$, respectively. This family thus provides the required homotopy.
    \end{proof}
    
    \subsection{Class $\rBDI$} 
    \label{ssec:BDI}
    
    We start from class $\rBDI$, characterized by even $T$- and $C$-symmetries. 
    
    \begin{lemma} \label{lem:normalForm_BDI}
        Assume $\rBDI(0, n, N)$ is non empty. Then $N = 2n$, and there is a basis of $\cH$ in which
        \begin{equation*} 
            T = \begin{pmatrix}
                K_n & 0 \\ 0 & K_n
            \end{pmatrix}, \quad
            C =  \begin{pmatrix}
                K_n & 0 \\ 0 & -K_n
            \end{pmatrix},
            \quad \text{so that} \quad
            S = TC = \begin{pmatrix}
                \bbI_n & 0 \\ 0 & -\bbI_n
            \end{pmatrix}.
        \end{equation*}
    \end{lemma}
    
    \begin{proof}
        Let $P_0 \in \rBDI(0, n, 2n)$, and let $\{\phi_1, \cdots, \phi_n\}$ be an orthonormal basis for ${\rm Ran} \, P_0$ such that $T \phi_j = \phi_j$ for all $1 \le j \le n$ (see Lemma~\ref{lem:formTeven}). We set
        \[
            \forall 1 \le j \le n, \quad \phi_{n+j} = C \phi_j.
        \]
        Since $C$ is anti-unitary, and maps ${\rm Ran} \, P_0$ into ${\rm Ran} \, (\bbI-P_0)$, the family $\{ \phi_1, \cdots, \phi_{2n}\}$ is an orthonormal basis for $\cH$. Since $T$ and $C$ commute, we have for all $1 \le j \le n$,
        \begin{equation} \label{eq:BDI_TC}
            T \phi_{n+j}  = T C \phi_j = C T \phi_j = C \phi_j = \phi_{n+j}, \quad \text{and} \quad 
            C \phi_{n+j} = C^2 \phi_j = \phi_j.
        \end{equation}
        Therefore in this basis the operators $T$ and $C$ take the form
        \[
            T = \begin{pmatrix}
                K_n & 0 \\ 0 & K_n
            \end{pmatrix}, \quad 
            C = \begin{pmatrix}
               0 & K_n \\ K_n & 0
            \end{pmatrix}
        \quad \text{and} \quad
        S = \begin{pmatrix}
            0 & \bbI_n \\ \bbI_n & 0
        \end{pmatrix}.
        \]
        We now change basis via the matrix $U := \frac{1}{\sqrt{2}} \begin{pmatrix}
            \bbI_n & \bbI_n \\ \bbI_n & - \bbI_n
        \end{pmatrix}$ to obtain the result.
    \end{proof}

    Using Lemma~\ref{lem:normalForm_BDI}, one can describe a projection $P(k)$ with its corresponding unitary $Q(k)$. The condition $T^{-1} P(k) T =  P(-k)$ reads
    \[
    \overline{Q}(-k) = Q(k).
    \]
    So a projection $P$ is in  $\rBDI(0, n , 2n)$ iff the corresponding matrix $Q \in \U(n)$ satisfies $\overline{Q} = Q$, that is $Q \in \O(n)$. This proves that
    \[
    \boxed{ \rBDI(0, n, 2n) \simeq \O(n). }
    \]
    Recall that $\O(n)$ has two connected components, namely $\det^{-1} \{ \pm 1\}$.

    \begin{theorem}[Class $\rBDI$] \label{th:BDI}
        The set $\rBDI(d, n, N)$ is non-empty iff $N=2n$. 
        \begin{itemize}
            \item Let $\Index_0^{\rBDI} \colon \rBDI(0,n,2n) \to \Z_2$ be the index map defined by
            \[
            \forall P \in \rBDI(0,n, 2n), \quad \Index_0^{\rBDI} (P) = \det(Q).
            \]
            Then $P_0$ is homotopic to $P_1$ in $\rBDI(0,n,2n)$ iff $\Index_0^{\rBDI}(P_0) = \Index_0^{\rBDI}(P_1)$.
            
            \item There is an index map $\Index_1^{\rBDI} \colon \rBDI(1,n,2n) \to \Z_2 \times \Z$ such that $P_0$ is homotopic to $P_1 $ in $\rBDI(1,n,2n)$ iff $\Index_1^{\rBDI}(P_0) = \Index_1^{\rBDI}(P_1)$.
        \end{itemize}
    \end{theorem}

    \begin{proof}
        Recall that $\SO(n)$ is path-connected, see the proof of Theorem~\ref{thm:AI}. The complement $\O(n) \setminus \SO(n)$ is in bijection with $\SO(n)$, by multiplying each orthogonal matrix with determinant $-1$ by the matrix ${\rm diag}(1,1,\ldots,1,-1)$. This proves the first part.
        
        \medskip
        
        We now focus on dimension $d = 1$. Let $P(k)$ be in $\rBDI(1, n, 2n)$, and let $Q(k)$ be the corresponding unitary. Let $\alpha(k) : [0, \tfrac12] \to \R$ be a continuous map so that
        \[
            \forall k \in [0, \tfrac12], \quad \det \, Q(k) = \re^{ \ri \alpha(k)}.
        \]
        Since $Q(0)$ and $Q(\tfrac12)$ are in $\O(n)$, we have $\det \, Q(0) \in \{ \pm 1\}$ and $\det \, Q(\tfrac12) \in \{ \pm 1\}$. We define
        \[
             \cW^{1/2} (P) :=  \cW^{1/2} (Q) := \dfrac{1}{\pi} \left( \alpha(\tfrac12) - \alpha(0) \right)  \quad \in \Z.
        \]
        The number $\cW^{1/2} (Q) \in \Z$ counts the number of {\em half turns} that the determinant is winding as $k$ goes from $0$ to $\tfrac12$. We call this map the {\em semi-winding}. We finally define the index map $\Index_1^{\rBDI} \colon \rBDI(1,n,2n) \to \Z_2 \times \Z$ by
        \[
            \forall P \in \rBDI(1, n, 2n), \quad \Index_1^{\rBDI}(P) := \left( \det \, Q(0), \ \cW^{1/2}(P) \right) \quad \in \Z_2 \times \Z.
        \]
        
        Let $P_0, P_1$ be in $\rBDI(1, n, 2n)$ such that $\Index_1^{\rBDI}(P_0) = \Index_1^{\rBDI}(P_1)$, and let us construct an homotopy between $P_0$ and $P_1$. First, we have $\det \, Q_0(0) = \det \, Q_1(0)$, and, since $\cW^{1/2}(P_0) = \cW^{1/2}(P_1)$, we also have $\det \, Q_0(\tfrac12) = \det \, Q_1(\tfrac12)$.
        
       Let $Q_s(0)$ be a path in $\O(n)$ connecting $Q_0(0)$ and $Q_1(0)$, and let $Q_s(\tfrac12)$ be a path connecting $Q_0(\tfrac12)$ and $Q_1(\tfrac12)$. This defines a continuous family of unitaries on the boundary of the half-square $\Omega_0 := [0, \tfrac12] \times [0,1]$. 
       %One can extend this family inside $\Omega_0$ iff the winding of this family $Q$ along this boundary vanishes. 
       Since $Q_s(0)$ and $Q_s(\tfrac12)$ are in $\O(n)$ for all $s$, their determinants are constant, equal to $\{ \pm 1\}$, and they do not contribute to the winding of the determinant of this unitary-valued map. So the winding along the boundary equals
        \[
        {\rm Winding}(\partial \Omega_0, Q) = \cW^{1/2}(P_0) - \cW^{1/2}(P_1) = 0.
        \]
    Lemma~\ref{lem:Ch1d} allows then to conclude the proof.
    \end{proof}

    %%%%%%%%%%%%%%%%%%%%%
    
    \subsection{Class $\rCI$}
    \label{ssec:CI}
    
    In class $\rCI$, the $T$-symmetry is even ($T^2 = \bbI_\cH$) while the $C$-symmetry is odd ($C^2 = - \bbI_\cH$).
    
    \begin{lemma}
        Assume $\rCI(0, n, N)$ is non empty. Then $N=2n$, and there is a basis of $\cH$ in which
        \begin{equation} \label{eq:form_CI}
            T = \begin{pmatrix}
                0 & K_n \\ K_n & 0
            \end{pmatrix}, \quad
            C = \begin{pmatrix}
                0 & -K_n \\ K_n & 0
            \end{pmatrix}
            \quad \text{so that} \quad
            S = TC = \begin{pmatrix}
                \bbI_n & 0 \\ 0 & - \bbI_n
            \end{pmatrix}.
        \end{equation}
    \end{lemma}

    \begin{proof}
        The proof is similar to the one of Lemma~\ref{lem:normalForm_BDI}. This time, since $C^2 = - \bbI$ and $TC = - CT$, we have, instead of~\eqref{eq:BDI_TC},
        \[
            T \phi_{n+j}= T C \phi_{j} = - CT \phi_{j} = - C \phi_j = - \phi_{n+j}, \quad \text{and} \quad
            C \phi_{n+j} = C^2 \phi_{j} = - \phi_j. \qedhere
        \]
    \end{proof}
    
    Using again Lemma~\ref{lem:normalForm_BDI}, we describe a projection $P(k)$ with its corresponding unitary $Q(k)$. The condition $T^{-1} P(k) T = P(-k)$ gives
    \[
    Q(-k)^T = Q(k).
    \]
    In particular, if $P \in \rCI(0, n, 2n)$, the corresponding $Q$ satisfies $Q^T = Q$. In Corollary~\ref{cor:AutonneTagaki} below, we prove that a matrix $Q$ is in $\U(n) \cap \cS_n^{\R}(\C)$ iff it is of the form
    \[
        Q = V^T V, \quad \text{for some} \quad V \in \U(n).
    \]
    In addition, we have $Q = V_0^T V_0 = V_1^T V_1$ with $V_0, V_1 \in \U(n)$ iff $V_0 V_1^* \in \O(n)$. This proves that
    \[
    \boxed{ \rCI(0, n, 2n) \simeq \U(n) \cap \cS_n^{\R}(\C) \simeq \U(n) / \O(n).}
    \]

    \begin{theorem}[Class $\rCI$]
        The set $\rCI(d, n, N)$ is non-empty iff $N=2n$. It is path-connected both for $d=0$ and for $d=1$.
    \end{theorem}

    \begin{proof}
        Given two matrices $Q_0$, $Q_1$ in $\U(n) \cap \cS_n^{\R}(\C)$, we can connect them in $\U(n) \cap \cS_n^{\R}(\C)$ by connecting the corresponding $V_0$ and $V_1$ in $\U(n)$. This proves that $\rCI(0, n, 2n)$ is connected.
        
        \medskip
        
        We now focus on the case $d = 1$. Let $P_0(k)$ and $P_1(k)$ be two families in $\rCI(1,n)$, with corresponding unitaries $Q_0$ and $Q_1$. Let $V_0(0), V_1(0) \in \U(n)$ so that
        \[
            Q_0(0) = V_0(0)^T V_0(0), \quad \text{and} \quad Q_1(0) = V_1(0)^T V_1(0).
        \]
        Let $V_s(0)$ be a homotopy between $V_0(0)$ and $V_1(0)$ in $\U(n)$, and set
        \[
            Q_s(0) := V_s(0)^T V_s(0).
        \]
        Then, $Q_s(0)$ is a homotopy between $Q_0(0)$ and $Q_1(0)$ in $\rCI(0,n, 2n)$. We construct similarly an homotopy between $Q_0(\tfrac12)$ and $Q_1(\tfrac12)$ in $\rCI(0,n, 2n)$. 
        This gives a path of unitaries on the boundary of the half-square $\Omega_0$. We can extend this family inside $\Omega_0$ iff the winding of the determinant along the boundary loop vanishes.
 
        \medskip
        
        Let $W \in \Z$ be this winding. There is no reason {\em a priori} to have $W = 0$. However, if $W \neq 0$, we claim that we can cure the winding by modifying the path $V_s(0)$ connecting $V_0(0)$ and $V_1(0)$. Indeed, setting
        \[
            \widetilde{V}_s(0) = {\rm diag} (\re^{ \ri W\pi s/2}, 1, 1, \cdots, 1) V_s(0), \quad \text{and} \quad \widetilde{Q}_s(0) := \widetilde{V}_s(0)^T \widetilde{V}_s(0),
        \]
        we can check that the family $\widetilde{Q}_s(0)$ also connects $Q_0(0)$ and $Q_1(0)$ in $\rCI(0,n ,2n)$, and satisfies
        \[
        \det \ \widetilde{Q}_s(0) =  \re^{ \ri W \pi s} \, \det \ Q_s(0).
        \]
        This cures the winding, and Lemma~\ref{lem:Ch1d} allows to conclude that the class $\rCI(1, n, 2n)$ is path-connected.
    \end{proof}

    %%%%%%%%%%%%%%%%%%%%%%%%%%%%%%%

    \subsection{Class $\rDIII$}
    \label{ssec:DIII}
    
    The class $\rDIII$ mirrors $\rCI$, since here the $T$-symmetry is odd ($T^2 = -\bbI_\cH$) while the $C$-symmetry is even ($C^2 = \bbI_\cH$). This class has been studied {\em e.g.} in~\cite{deNittis2021cohomology}.
    \begin{lemma} \label{lem:normalForm_DIII}
        Assume $\rDIII(0, n, N)$ is non empty. Then $n = 2m$ is even, and $N = 2n = 4m$ is a multiple of $4$. There is a basis of $\cH$ in which
        \[
        T = \begin{pmatrix}
            0 & K_n J_n \\ K_n J_n & 0
        \end{pmatrix}, 
    \quad
        C = \begin{pmatrix}
            0 & K_n J_n \\ - K_n J_n & 0
        \end{pmatrix}, \quad \text{and} \quad
    S = \begin{pmatrix}
        \bbI_{n} & 0  \\ 0 & -\bbI_{n}
    \end{pmatrix}.
        \]
    \end{lemma}
    \begin{proof}
        Let $P_0 \in \rDIII(0, n, 2n)$. Since $T$ is anti-unitary, leaves ${\rm Ran} \, P_0$ invariant, and satisfies $T^2 = - \bbI_{{\rm Ran} \, P_0}$ there, one can apply Lemma~\ref{lem:formTodd} to the restriction of $T$ on ${\rm Ran} \, P_0$. We first deduce that $n = 2m$ is even, and that there is a basis for $\Ran \, P_0$ of the form $\set{\psi_1, \ldots, \psi_{2m}}$, with $\psi_{m+j} = T \psi_j$. Once again we set $\psi_{2m+j} := C \psi_j$. This time, we have $TC = - CT$, so, in the basis $\set{\psi_1, \ldots, \psi_{4m}}$, we have
        \[ 
            T = \begin{pmatrix} K J_n & 0 \\ 0 & -K J_n \end{pmatrix}, \quad 
            C= \begin{pmatrix} 0 & K \\ K & 0 \end{pmatrix} \quad \text{hence} \quad 
            S=TC = \begin{pmatrix} 0 & J_n \\ -J_n & 0 \end{pmatrix}, \]
        A computation reveals that
       \[ 
        U^* \begin{pmatrix} 0 & J_n \\ -J_n & 0 \end{pmatrix} U = \begin{pmatrix} \bbI_n & 0 \\ 0 & - \bbI_n \end{pmatrix}, \quad \text{with} \quad U := \frac{1}{\sqrt{2}} \begin{pmatrix} \bbI_m & 0 & -\bbI_m & 0 \\ 0 & - \bbI_m & 0 & \bbI_m \\ 0 & \bbI_m & 0 & \bbI_m \\ \bbI_m & 0 & \bbI_m & 0  \end{pmatrix},
    \]
    and that $U$ is unitary. With this change of basis, we obtain the result.
    \end{proof}

    In this basis, we have that $T^{-1} P(k) T = P(-k)$ iff the corresponding $Q$ satisfies
    \[
        J_n Q^T(-k) J_n = - Q(k).
    \]
     In dimension $d=0$, the condition becomes $J_n Q^T J_n = - Q$, which can be equivalently rewritten as
    \[ 
        A^T = - A, \quad \text{with} \quad A := Q J_n. 
    \]
    The matrix $A$ is unitary and skew-symmetric, $A \in \U(n) \cap \cA_{n}^\R(\C)$. In particular, the Pfaffian of $A$ is well-defined. In Corollary~\ref{cor:O(d)capA(d)} below, we recall that a matrix $A$ is in $\U(n) \cap \cA_n^\R(\C)$ iff it is of the form
    \[
        A = V^T J_n V, \quad \text{with} \quad V \in \U(n).
    \]
    In addition, we have $A = V_0^T J_n V_0 = V_1^T J_n V_1$ with $V_0, V_1 \in \U(n)$ iff $V_0 V_1^* \in \Sp(m)$. Therefore
    \[
        \boxed{ \rDIII(0, 2m, 4m) \simeq   \U(2m) \cap \cA_{2m}^\R(\C) \simeq \U(2m) /  \Sp(m). }
    \]
    
    \begin{theorem}[Class $\rDIII$]
        The set $\rDIII(d, n, N)$ is non-empty iff $n=2m \in 2\N$ and $N=2n=4m$. 
        \begin{itemize}
            \item The set $\rDIII(0,2m,4m)$ is path-connected.
            \item There is a map $\Index_1^{\rDIII} \colon \rDIII(1,2m,4m) \to \Z_2$ such that $P_0$ is homotopic to $P_1$ in $\rDIII(1,2m,4m)$ iff  $\Index_1^{\rDIII}(P_0) = \Index_1^{\rDIII}(P_1)$.
        \end{itemize}
    \end{theorem}
    The index $\Index_1^{\rDIII}$ is defined below in~\eqref{eq:def:Index_DIII_1}. It matches the usual Teo-Kene formula in~\cite[Eqn. (4.27)]{Teo_2010}.
    \begin{proof}
        For the first part, it is enough to connect the corresponding matrices $V$'s in $\U(n)$, which is path-connected. 
        
        Let us focus on the case $d = 1$. Let $P(k) \in \rDIII(1, 2m, 4m)$ with corresponding matrices $Q(k) \in \U(2m)$ and $A(k) := J_{2m}^T Q(k) \in \U(2m) \cap \mathcal{A}^\R_{2m}(\C)$. Let $\alpha(k) : [0, 1/2] \to \R$ be a continuous phase so that
    \[
    \forall k \in [0, \tfrac12], \quad \det \ A(k) = \re^{ \ri \alpha(k)}.
    \]
    For $k_0 \in \{ 0, 1/2 \}$, $A(k_0)$ is anti-symmetric, so we can define its Pfaffian, which satisfies $\Pf( A(k_0))^2 = \det \ A(k_0) = \re^{ \ri \alpha(k_0)}$. Taking square roots shows that there are signs $\sigma_0, \sigma_{1/2} \in \{ \pm 1 \}$ so that
    \[
    \Pf \ A(0) = \sigma_0 \re^{ \ri \tfrac12 \alpha(0)}, \quad \text{and} \quad
    \Pf \ A(\tfrac12) = \sigma_{1/2} \re^{ \ri \tfrac12 \alpha(\tfrac12)}.
    \]
    We define the Index as the product of the two signs $\sigma_0 \cdot \sigma_{1/2}$. Explicitly,
    \begin{equation} \label{eq:def:Index_DIII_1}
    \Index_1^\rDIII(P) := \dfrac{\re^{ \ri \tfrac12 \alpha(0)}} {\Pf \, A (0)} \cdot  \dfrac{\re^{ \ri \tfrac12 \alpha(\frac12)}} {\Pf \, A (\frac12)} \quad \in \{ \pm 1\}.
    \end{equation}
    Note that this index is independent of the choice of the lifting $\alpha(k)$. Actually, this index is $1$ if, by following the continuous map $\re^{\ri \tfrac12 \alpha(k)}$, that is a continuous representation of $\sqrt{\det(A(k))}$, one goes from $\Pf \, A(0)$ to $\Pf \, A(\tfrac12)$, and is $-1$ if one goes from $\Pf \, A(0)$ to $-\Pf \, A(\tfrac12)$.
    
    Let us prove that if $P_0, P_1 \in \rDIII(1, 2m, 4m)$, then $\Index_1^\rDIII(P_0) = \Index_1^\rDIII(P_1)$ iff there is an homotopy between the two maps. Let $V_0(0), V_1(0) \in \U(n)$ be so that
    \[
         A_0(0) = V_0(0)^T J_n V_0(0), \quad \text{and} \quad  A_1(0) = V_1(0)^T J_n V_1(0).
    \]
    Let $V_s(0)$ be a homotopy between $V_0(0)$ and $V_1(0)$ in $\U(n)$, and set 
    \[
        A_s(0) := V_s(0)^T J_n V_s(0).
    \]
    This gives a homotopy between $A_0(0)$ and $A_1(0)$ in $\rDIII(0, n, 2n)$. We construct similarly a path $A_s(\tfrac12)$ connecting $A_0(\tfrac12)$ and $A_1(\tfrac12)$ in $\rDIII(0, n, 2n)$. 

    Define continuous phase maps $\alpha_0(k)$, $\widetilde{\alpha_s}(\tfrac12)$, $\alpha_1(k)$ and $\widetilde{\alpha_s}(0)$ so that
        \[
        \forall k \in [0, \tfrac12], \quad \det \ A_0(k) = \re^{ \ri \alpha_0(k)} \quad \text{and} \quad
        \det \ A_1(k) = \re^{ \ri \alpha_1(k)},
        \]
        while
        \[
        \forall s \in [0, 1], \quad \det \ A_s(0) = \re^{ \ri \widetilde{\alpha_s}(0)} \quad \text{and} \quad
        \det \ A_s(\tfrac12) = \re^{ \ri \widetilde{\alpha_s}(\tfrac12)},
        \]
        together with the continuity conditions
        \[
        \alpha_0(k=\tfrac12) = \widetilde{\alpha_{s = 0}}(\tfrac12), \quad
        \widetilde{\alpha_{s=1}}(\tfrac12) = \alpha_1(k = \tfrac12), \quad \text{and} \quad
        \alpha_1(k=0) = \widetilde{\alpha_{s = 1}}(0).
        \]
        With such a choice, the winding of $\det(A)$ along the loop $\partial \Omega_0$ is
        \[
        W := \dfrac{1}{2 \pi} \left[ \widetilde{\alpha_0}(0) - \alpha_0(0) \right] \in \Z.
        \]
        We claim that $W \in 2 \Z$ is even iff $\Index_1^\rDIII(P_0) = \Index_1^\rDIII(P_1)$. The idea is to follow a continuation of the phase of $\sqrt{ \det \, A}$ along the boundary. For $j \in \{ 0, 1\}$, we denote by $\eps_j := \Index_1^\rDIII(P_j)$ the index for the sake of clarity.
        
        By definition of the Index, we have 
        \[
            \dfrac{ \re^{ \ri \frac12 \alpha_0(\tfrac12)} } { \Pf \, A_0(\tfrac12)} =  \dfrac{ \re^{ \ri \frac12 \alpha_0(0)} } { \Pf \, A_0(0) } \, \eps_0, \quad \text{and, similarly, } \quad
            \dfrac{ \re^{ \ri \frac12 \alpha_1(\tfrac12)} } { \Pf \, A_1(\tfrac12)} =  \dfrac{ \re^{ \ri \frac12 \alpha_1(0)} } { \Pf \, A_1(0) } \, \eps_1
        \]
        On the segment $(k,s) = \{ \tfrac12 \} \times [0, 1]$, the map $s \mapsto \Pf \, A_s(\tfrac12)$ is continuous, and is a continuous representation of the square root of the determinant. So
        \[
            \dfrac{\re^{ \ri \tfrac12 \widetilde{\alpha_0}(\tfrac12)} }{\Pf \, A_0(\tfrac12)} =  
            \dfrac{\re^{ \ri \tfrac12 \widetilde{\alpha_1}(\tfrac12)} }{\Pf \, A_1(\tfrac12)} ,
            \quad \text{and similarly,} \quad
             \dfrac{\re^{ \ri \tfrac12 \widetilde{\alpha_0}(0)} }{\Pf \, A_0(0)} =  
            \dfrac{\re^{ \ri \tfrac12 \widetilde{\alpha_1}(0)} }{\Pf \, A_1(0)}.
        \]
        Gathering all expressions, and recalling the continuity conditions, we obtain
        \[
             \dfrac{ \re^{ \ri \tfrac12 \widetilde{\alpha_0}(0)} } { \Pf \, A_0(0) } = \eps_0 \eps_1  \dfrac{ \re^{ \ri \tfrac12 \alpha_0(0)} } { \Pf \, A_0(0) }, \quad \text{so} \quad \re^{ \ri \pi W} = \eps_0 \eps_1.
        \]
        This proves our claim.

        If the indices differ, then we have $\eps_0 \eps_1 = - 1$, hence $W$ is odd. In particular, $W \neq 0$, and one cannot find an homotopy in this case. Assume now that that two indices are equal, so that $\eps_0 \eps_1 = 1$ and $W \in 2 \Z$ is even. There is no reason {\em a priori} to have $W = 0$, but we can cure the winding. Indeed, we set 
        \[
        \widetilde{A}_s(0) := \widetilde{V}_s(0)^T J_n \widetilde{V}_s(0), \quad \text{with} \quad
        \widetilde{V}_s(0) := {\rm diag} ( \re^{ \ri \pi W s}, 1, \cdots, 1) V_s(0).
        \]
        The family $A_s(0)$ is a continuous family connecting $A_0(0)$ and $A_1(0)$ in $\rDIII(0, n, 2n)$. In addition, we have $\det \, \widetilde{A}_s(0) = \re^{ 2 \ri \pi W s } \det \, A_s(0)$, so this new interpolation cures the winding. Invoking Lemma~\ref{lem:Ch1d} concludes the proof.
    \end{proof}

    \subsection{Class $\rCII$}
    \label{ssec:CII}
    
    Finally, it remains to study the class $\rCII$, in which we have both $T^2 = -\bbI_\cH$ and $C^2 = -\bbI_\cH$. 
    \begin{lemma}
         Assume $\rCII(0, n, N)$ is non empty. Then $n = 2m$ is even, and $N = 2n = 4m$ is a multiple of $4$. There is a basis of $\cH$ in which
        \[
        T = \begin{pmatrix}
             -K_n J_n & 0 \\ 0 & - K_n J_n
        \end{pmatrix}, 
        \quad
        C = \begin{pmatrix}
            K_n J_n & 0 \\ 0 & - K_n J_n
        \end{pmatrix}, \quad \text{and} \quad
        S = \begin{pmatrix}
            \bbI_{n} & 0  \\ 0 & -\bbI_{n}
        \end{pmatrix}.
        \]
    \end{lemma}
\begin{proof}
    The proof is similar to the one of Lemma~\ref{lem:normalForm_DIII}. Details are left to the reader.
\end{proof}
    In this basis, the condition $T^{-1} P(k) T = P(-k)$ reads, in terms of $Q$,
    \[
        J_n \overline{Q}(k) J_n = - Q(-k), \quad \text{or equivalently} \quad Q(k)^T J_n Q(-k) = J_n.
    \]
    In particular, in dimension $d = 0$, we have $Q \in \U(2m) \cap \Sp(2m; C) = \Sp(m)$. So
    \[
        \boxed{ \rCII(0, 2m, 4m) \simeq \Sp(m).}
    \]
    
    \begin{theorem}[Class $\rCII$]
        The set $\rCII(d, n, N)$ is non-empty iff $n=2m \in 2\N$ and $N=2n=4m$. 
        \begin{itemize}
            \item The set $\rCII(0,2m,4m)$ is path-connected. 
            \item Define the map $\Index_1^{\rCII} \colon \rCII(1,2m,4m) \to \Z$ by
            \[
                \forall P \in \rCII(1, 2m, 4m), \quad \Index_1^{\rCII}(P) := {\rm Winding}(\TT^1, Q).
            \]
            Then $P_0$ is homotopic to $P_1$ in $\rCII(1,2m,4m)$ iff $\Index_1^{\rCII}(P_0) = \Index_1^{\rCII}(P_1)$.
        \end{itemize}
    \end{theorem}
    \begin{proof}
        We already proved in Theorem~\ref{thm:AII} that $\Sp(m)$ is connected, which yields the first part. 
        
        For the $d = 1$ case, we first note that if $Q \in \Sp(m)$, we have $Q^T J_n Q = J_n$. Taking Pfaffians, we get $\det(Q) = 1$. As in the proof of Theorem~\ref{th:BDI}, we deduce that any path $Q_s(0)$ connecting $Q_0(0)$ and $Q_1(0)$ in $\Sp(m)$ has a determinant constant equal to $1$, hence does not contribute to the winding. The proof is then similar to the one of Theorem~\ref{th:BDI}.
    \end{proof}

    %%%%%%%%%%%%%%%%%%%%%%%%%%%%%%%%%%%%%
    
    \appendix

    \section{Matrix factorizations} 
    \label{sec:LA}
    
    In this appendix, we show how to factorize certain classes of matrices we encountered in the main body of the paper. The first result has been discovered many times, and is known as the Autonne--Tagaki factorization~\cite{autonne1915matrices}. The proof we present is found in \cite[Cor.~4.4.4]{Horn} for the complex case, and in \cite{FassbenderBig} for the symplectic case. For the sake of the reader, we give a unified proof which employs also tools from \cite{FassbenderSmall} in the symplectic case.
    
    Recall that we denote by $\cS_{n}^\R(\C)$ the set of $n \times n$ (complex) matrices, symmetric in the sense $A = A^T$, and by $\cA_n^\R(\C)$ the anti-symmetric ones, satisfying $A = -A^T$.
    
    \begin{theorem}[Autonne--Tagaki factorization] \label{th:AutonneTagaki}
        Let $A \in \cS_n^\R(\C)$, and let $\Lambda$ be the diagonal matrix composed of the (non-negative) singular values of $A$. Then there is a unitary $U \in \U(n)$ such that $A = U \Lambda U^T$.
        
        If $n=2m$ and $A$ is also symplectic, i.e.\ $A^T J_n A = J_n$, then $U$ and $\Lambda$ can be chosen to be symplectic as well.
    \end{theorem}
    
    The above factorization is \emph{not} a spectral decomposition, which involves \emph{similarities} of the form $U^* A U$ rather than \emph{congruences} of the form $U^T A U$. 
    
    \begin{proof}
        By definition, singular values of $A$ are the (non-negative) square roots of the eigenvalues of the positive Hermitian operator $H := A^* A$. The operator $H$ is hermitian, hence diagonalizable, of the form $H = W \Lambda^2 W^*$ for some $W \in \U(n)$. Define
        \[ L := W^T A W \]
        and observe that $L^T = L$. Using the unitarity of $W$ and the symmetry of $A$, we have
        \begin{align*}
            L^* L & = W^* A^* \overline{W} W^T A W = W^* A^* A W = \Lambda^2, \\
            L L^* & = W^T A W W^* A^* \overline{W}  = W^T A \overline{A} \overline{W} = \overline{W^* A^* A W} = \overline{\Lambda^2} = \Lambda^2,
        \end{align*}
        as $\Lambda^2$ is real-valued. So $L L^* = L^* L$, {\em i.e.} the operator $L$ is normal, thus admits a polar decomposition $L = V P$ with $P := (L^* L)^{1/2} = \Lambda$ and $V \in \U(n)$ which \emph{commute} among each other. As $\Lambda$ is diagonal and $V$ commutes with it, we may choose $V$ to be diagonal as well. In particular, it is symmetric, $V = V^T$. If $V = {\rm diag}(\re^{\ri \phi_1}, \ldots, \re^{\ri \phi_n})$, denote by $V^{1/2} := {\rm diag}(\re^{\ri \phi_1/2}, \ldots, \re^{\ri \phi_n/2})$. This gives
        \[ W^T A W = L = V \Lambda = V^{1/2} \Lambda V^{1/2} \quad \text{hence} \quad A = U \Lambda U^T \quad \text{with} \quad U := \overline{W} V^{1/2} . \]
        
        In the symplectic case, we can prove that each matrix appearing previously can be chosen symplectic as well. For the matrices $W$ and $\Lambda^2$, this follows from~\cite[Prop.~3]{FassbenderSmall}. We directly check that $L$ is symplectic if $A$ is, and its polar decomposition $L = VP$ can be chosen with $V$ and $P$ symplectic~\cite[Prop.~4]{FassbenderSmall}.
    \end{proof}

    We deduce the following useful corollary. The fact that $\Sp(m) \cap \O(2m) \simeq \U(m)$ is proved below in Proposition~\ref{prop:SpO=U}.
    
    \begin{corollary} \label{cor:AutonneTagaki} ~\\
        $\bullet$ A matrix $A$ is in $\U(n) \cap \cS_n^\R(\C)$ iff it is of the form 
        \[
            A = V^T V \quad \text{with} \quad V \in \U(n).
        \]
        In addition, $A = V_0^T V_0 = V_1^T V_1$ with $V_0, V_1 \in \U(n)$ iff $V_0 V_1^* \in \O(n)$. In particular, $\U(n) \cap \cS_n^\R(\C) \simeq \U(n) / \O(n)$.
        
        \smallskip
        
        \noindent $\bullet$ A matrix $A$ is in $\O(n) \cap \cS_n^\R(\C)$ iff there is $0 \le j \le n$ so that $A$ is of the form
        \[
            A = V^T D_j V \quad \text{with} \quad V \in \SO(n), \quad \text{and} \quad D_j := {\rm diag} ( \underbrace{1, \cdots, 1}_{j}, \underbrace{-1, \cdots, -1}_{n-j} ).
        \]
        The set $\O(n) \cap \cS_n^\R(\C)$ has $n+1$ connected components, labeled by the signature.
        \smallskip
        
        \noindent $\bullet$ Assume $n = 2m$. A matrix $A$  is in $\Sp(m) \cap \cS_{2m}^\R(\C)$ iff it is of the form 
         \[
        A = V^T V \quad \text{with} \quad V \in \Sp(m).
        \]
       In addition, $A = V_0^T V_0 = V_1^T V_1$ with $V_0, V_1 \in \Sp(m)$ iff $V_0 V_1^* \in \Sp(m) \cap \O(2m) \simeq \U(m)$. In particular, $\Sp(m) \cap \cS_{2m}^\R(\C) \simeq \Sp(m) / \U(m)$.
    \end{corollary}
    
    \begin{proof}
        If $A \in \U(n) \cap \cS_n^\R(\C)$, its singular values are all equal to $1$, so the Autonne--Tagaki factorization of $A$ is of the form $A= U U^T = V^T V$ with $V = U^T$. \\
        If $A = V_0^T V_0 = V_1^T V_1$ with $V_0, V_1 \in \U(n)$, then $Z := V_0V_1^*$ is unitary, and satisfies $Z = (Z^*)^T = \overline{Z}$, that is $Z$ is real-valued. So $Z \in \O(n)$ as wanted.\\
        The proof in the symplectic case is similar. \\
        If $A \in \O(n) \cap \cS_n^\R(\C)$, then the usual diagonalization of $A$ shows that $A$ is of the form $A = V^T D V$ with $V \in \SO(n)$, and $D$ the diagonal matrix with the eigenvalues of $A$, counting multiplicities, and ranked in decreasing order. Since the eigenvalues of $A \in \O(n)$ are $\pm 1$, we have $D = D_j$ with $j = \dim \, \Ker(A - 1)$. If $A_0$ and $A_1$ have different signatures, they are in different connected components of $\O(n) \cap \cS_N^\R(\C)$, and if they have the same signature, they can be connected by connecting the corresponding $V$'s. This concludes the proof.
    \end{proof}
    
    %%%%%%%%%%%%%%%%%%%%%%%%%%%%%%
    
    The anti-symmetric analogue of the Autonne--Tagaki factorization is known as Hua's decomposition~\cite[Thm.~7]{Hua}. We set $J_2 = \begin{pmatrix}
        0 & 1 \\ - 1 & 0
    \end{pmatrix}$.
    
    \begin{theorem}[Hua's decomposition] \label{th:Hua}
        Let $A \in \cA_n^{\R}(\C)$ and assume that $A$ is non-degenerate. Then $n = 2m$ and all singular values of $A$ have even multiplicities. In addition, if $D$ is the $m \times m$ diagonal matrix, composed of ``half'' the singular values of $A$, then there is $U \in \U(n)$ such that $A = U \left( D \otimes J_2 \right) U^T$. \\
        If $A$ is real-valued, then $U$ is orthogonal.
    \end{theorem}

    Recall that if $n$ is odd, then $A = -A^T$ is never non-degenerate, since $\det(A) = \det(A^T) = (-1)^n \det(A) = - \det(A)$, hence $\det(A) = 0$.
    
    \begin{proof}
        Consider the positive Hermitian matrix $H := A^* A$. Observe that, in view of the non-degeneracy of $A$,
        \begin{align*}
           \det \left( A^*A - \lambda \bbI_n \right) = \det\left( A^* - \lambda A^{-1} \right) \det(A) = \Pf(A)^2 \, \left[\Pf\left( A^* - \lambda A^{-1} \right) \right]^2,
        \end{align*}
        where we used that $A^* - \lambda A^{-1} \in \cA_n^\R(\C)$. Therefore, the characteristic polynomial of $H$ is a perfect square. This implies that its roots ({\em i.e.} the singular values of $A$) are of even multiplicity, and in particular their number, which is $n$, must be even.
        
        Consider now $\Lambda := D \otimes J_2$ as in the statement. If $D = {\rm diag}(\lambda_1, \ldots, \lambda_m)$, then
        \[ \Lambda^* \Lambda = \Lambda^T \Lambda = {\rm diag}(\lambda_1^2, \lambda_1^2, \ldots, \lambda_m^2, \lambda_m^2). \]
        The spectral decomposition of $H$ implies the existence of $W \in \U(m)$ such that $H = W \Lambda^* \Lambda W^*$. As in the proof of Theorem~\ref{th:AutonneTagaki}, set
        \[ L := W^T A W. \]
        Using that $A^T = - A$, we have $L^T = - L$. We also have $L^* L = \Lambda^* \Lambda$, which can be recast as
        \[ (L \Lambda^{-1})^{-1} = \Lambda L^{-1} = (\Lambda^*)^{-1} L^* = (L \Lambda^{-1})^*, \]
        that is, the matrix $L \Lambda^{-1} =: V$ is unitary. The skew-symmetry of $L$ and of $\Lambda$ also gives
        \[ V \Lambda = L = - L^T = - \Lambda^T V^T = \Lambda V^T. \]
        Diagonalize now the unitary matrix $V$:  
        \[ V = \Gamma \Delta \Gamma^{*} \quad \text{with} \quad \Gamma \in \U(n), \quad \Delta := {\rm diag}\left(\re^{\ri \phi_1}, \ldots, \re^{\ri \phi_n}\right). \]
        Set also
        \[ \Delta^{1/2} := {\rm diag}\left(\re^{\ri \phi_1/2}, \ldots, \re^{\ri \phi_n/2}\right), \quad V^{1/2} := \Gamma \Delta^{1/2} \Gamma^{*}. \]
        Clearly, $(V^{1/2})^2 = V$ and, by functional calculus, $V^{1/2} \Lambda = \Lambda (V^{1/2})^T$. We obtain
        \[ 
            W^T A W = L = V \Lambda = V^{1/2} \Lambda (V^{1/2})^T
        \]
        that is 
        \[ 
            A = U \left(  D \otimes J_2 \right) U^T, \quad \text{with} \quad U := (W^T)^* V^{1/2}. 
        \]
        This concludes the proof for the complex case. The real case follows immediately, upon noticing that a real-valued unitary matrix is automatically orthogonal.
    \end{proof}

    \begin{corollary} \label{cor:O(d)capA(d)} ~\\
    $\bullet $ A matrix $A$ is in $\U(2m) \cap \cA_{2m}^\R(\C)$ iff it is of the form 
    \[
        A = V^T J_{2m} V \quad \text{with} \quad V \in \U(2m).
    \]
     In addition, $A = V_0^T J_{2m} V_0 = V_1^T J_{2m} V_1$ with $V_0, V_1 \in \U(2m)$ iff $V_0V_1^* \in \Sp(m)$. In particular, $\U(2m) \cap \cA_{2m}^\R(\C) \simeq \U(2m) / \Sp(m)$. \\
     
    \noindent $\bullet$ A matrix $A$ is in $\O(2m) \cap \cA_{2m}^\R(\C)$ iff it is of the form
     \begin{equation} \label{eq:WAW=J}
        A = W^T J_{2m} W \quad \text{with} \quad W \in \O(2m).
    \end{equation} 
    and, setting $\Lambda_A := {\rm diag}(1, 1, \cdots , 1, \Pf(A)) \otimes J_2$, iff it is also of the form
    \begin{equation} \label{eq:VAV=Pf(A)}
        A = V^T \Lambda_A V \quad \text{with} \quad V \in \SO(2m).
    \end{equation}
    If $A = W_0^T J_{2m} W_0 = W_1^T J_{2m} W_1$ with $W_0, W_1 \in \O(2m)$, then $W_0 W_1^*  \in \Sp(2m;\R) \cap \O(2m) \simeq \U(m)$. In particular, $\O(2m) \cap \cA_{2m}^\R(\C) \simeq \O(2m) / \U(m)$.
    \end{corollary}
    \begin{proof}
        Let $A \in \U(2m) \cap \cA_{2m}^\R(\C)$. First, since $A \in \U(2m)$, it is invertible, and its singular values are all equal to $1$. The previous result gives $U \in \U(2m)$ so that
        \begin{equation} \label{eq:UAU=J}
            A = U \left( \bbI_m \otimes J_2 \right) U^T.
        \end{equation}
        Upon reshuffling the columns of $\bbI_m \otimes J_2$, we can transform $\bbI_m \otimes J_2$ into $J_{2m}$. We deduce that there is $V \in \U(2m)$ so that $A = V^T J_{2m} V$. If $A = V_0^T J_{2m} V_0 = V_1^T J_{2m} V_1$, with $V_0, V_1 \in \U(2m)$, then $Z := V_0 V_1^*$ is in $\U(2m)$, and satisfies $Z^T J_{2m} Z = J_{2m}$, that is, $Z$ is symplectic.
        
        The proof for~\eqref{eq:WAW=J} is similar. It remains to prove~\eqref{eq:VAV=Pf(A)}. Let $A \in \O(2m) \cap \cA_{2m}^\R(\C)$. Using~\eqref{eq:UAU=J} in the real case, we see that there is $U \in \O(2m)$ so that $A = U^T (\bbI_n \otimes J_2) U$ with $U \in \O(2m)$. We have $\det(U) \in \{ \pm 1\}$, and
        \[ 
            \Pf(A) = \Pf (U^T (\bbI_n \otimes J_2) U) = \det(U) \Pf((\bbI_n \otimes J_2)) = \det(U).
         \]
        If $\det(U) = 1$, we simply take $V := U$, otherwise we take
        \[
        V =  {\rm diag}(1, 1, \cdots, 1, -1) U. %\qedhere
        \]
    \end{proof}
    %%%%%%%%%%%%%%%%%%%%%%%%%%%%%
    
    We end this Appendix with the following result.
    
    \begin{proposition} \label{prop:SpO=U}
        The group $\Sp(2m;\R) \cap \O(2m) = \Sp(2m;\C) \cap \O(2m) \subset \Sp(m)$ is isomorphic to $\U(m)$.
    \end{proposition}
    \begin{proof}
        If $U \in \Sp(2m;\R) \cap \O(2m)$, then $J_{2m} = U^T J_{2m} U = U^{-1} J_{2m} U$ hence $U J_{2m} = J_{2m} U$. Decomposing $U$ in block form, this condition is equivalent to
        \[ 
        U = \begin{pmatrix} A & B \\ - B & A \end{pmatrix} = \overline{U} \quad \text{with} \quad A^T A + B^T B = \bbI_m, \quad A^T B - B^T A = 0. 
        \]
        These conditions are equivalent to the fact that the $m \times m$ matrix $V := A + \ri B$ is unitary. Indeed
        \[ (A+\ri B)^* (A+\ri B) = (A^T - \ri B^T)(A + \ri B) = (A^T A + B^T B) + \ri (A^T B - B^T A) = \bbI_m. \]
        This concludes the proof.
    \end{proof}

    \bibliographystyle{my-alpha}
    \bibliography{biblio}
    
\end{document}